\documentclass[a4paper,12pt]{article}

\usepackage{amsmath,amsthm}  
\usepackage{amssymb,latexsym}  
\usepackage{natbib}
\usepackage{graphicx}
\usepackage{xspace}
\usepackage{times}
\usepackage{setspace}
\usepackage{vmargin}

\theoremstyle{plain}
\newtheorem{theorem}{Theorem}
\newtheorem{lemma}[theorem]{Lemma}

\newtheorem{proposition}[theorem]{Proposition}

\theoremstyle{definition}
\newtheorem{definition}[theorem]{Definition}

\theoremstyle{remark}
\newtheorem{remark}[theorem]{Remark}

\numberwithin{theorem}{section}

\newcommand{\Z}{\mathbb Z}

\renewcommand{\P}{\mathbb P}

\newcommand{\E}{\mathbb E}

\newcommand{\pml}{\text{\sc{PML}}}

\newcommand{\snp}{\texttt{SNP}\xspace}
\newcommand{\snps}{\texttt{SNP}s\xspace}
\newcommand{\narac}{\texttt{NARAC}\xspace}
\newcommand{\maf}{\texttt{MAF}\xspace}
\newcommand{\hwe}{\texttt{HWE}\xspace}

\newcommand{\talque}{\enspace | \enspace}
\newcommand{\tab}{\enspace}

\setmargins{3cm}{3cm}{15.3cm}{23cm}{0cm}{0cm}{1cm}{1.5cm}

\begin{document}

\doublespacing

\title{Finding the basic neighborhood  in variable range Markov random fields:
application in  SNP association studies}
\author{Andr\'e J. Bianchi, Suely R. Giolo\\J\'ulia P. Soler and Florencia G. Leonardi}
%
%
%
%
%
\date{January 10, 2013}

\maketitle

\begin{abstract}

The \snps (Single Nucleotide Polymorphisms) genotyping platforms are of great value for gene mapping of complex diseases. Nowadays, the high-density of these molecular markers enables studies of dependence patterns between loci over the genome, allowing a simultaneous inference of dependence structure and disease association.
In this paper we propose a method based on the theory of variable range Markov random fields to estimate the extent of dependence among \snps allowing variable windows along the genome. The advantage of this method is that it allows the simultaneous prediction of dependence and independence regions among \snps, without restricting \emph{a priori} the range of dependence.
We introduce an estimator based on the idea of penalized maximum likelihood to find the conditional dependence neighborhood of each \snp in the sample and we prove its consistency.
We apply our method to autosomal \snps genotypic data with unknown phase in the context of case-control association studies.
By examining rheumatoid arthritis data from the \emph{Genetic Analysis Workshop 16} (GAW16), we show the utility of the Markov model under variable range dependence.
\end{abstract}


\section{Introduction}

Genome-wide genetic mapping studies based on linkage disequilibrium (LD) have been encouraged from the availability of high-density \snps (Single Nucleotide Polymorphisms) genotyping platforms. Since the individual \snp effect is expected to be small, a challenge has been to find  blocks of \snps with effect on diseases. This problem requires
the consideration of   a model of the dependence structure among loci over the genome, see for example    \citet{akey2001,greenspan2005,browning2006,kim2008}.

In this paper we propose a method based on the theory of variable range Markov random fields to estimate the extent of dependence among \snps allowing variable windows along the genome.
In order to infer the range of dependence for each \snp, we propose a criterion based on penalized maximum likelihood.
The main theoretical contribution  of this paper is the proof of the consistency of this estimator, as the sample size diverges.

As a consequence of our model, we show how the dependence information inferred can be used to construct independent blocks of \snps that can be associated  with a response variable in the context of  case-control studies.
Tools such as the Chi-square statistic are adopted to study the association of each block and the response variable. The main advantage of this method is that it allows, without restriction of range, the simultaneous prediction of dependence and independence regions among \snps.

Methods based on penalized  maximum likelihood to infer the range of dependence of Markov random fields have  been proposed in the literature before, see for example \citet{csiszar2006b,orlandi2011}. The main difference between these approaches and our is that these methods assume a fixed (and symmetric) neighborhood for each variable, because they based the inference procedure on only one sample of the process.
In our context, we have at hand several independent realizations of the  process, corresponding to different individuals in the sample,  and this allows us to assume an inhomogeneous process with different neighborhoods for each variable.

In order to illustrate the potential of our method to infer variable dependence structures and disease association we will consider rheumatoid arthritis data from the GAW16.  Although our methodology is driven to \snp genotypic sequence, extension for allele sequences is also possible.

The paper is organized as follows. In Section~\ref{sec:model} we define the variable range Markov random field, introduce the criterion to estimate the range of dependence for each variable and state the consistency result.
In Section~\ref{sec:snps} we show how to apply our method to infer the dependence structure in \snps maps
and to study association with a response variable in a case-control study. Finally, in Section~\ref{sec:proof}
we prove the main theoretical result of this paper.

\section{Variable range Markov random fields}\label{sec:model}

Let $A$ denote a finite alphabet and let
$\P$ be a probability distribution over $A^\Z$, equiped with the usual $\sigma$-algebra
generated by the cylinder sets.
Given an element $i\in\Z$, called a \emph{site}, we will denote by $X_i$
the marginal random variable obtained by the canonical projection of elements in $A^\Z$.
Given two non-negative integers $l$ and $r$, we will denote by $\Delta^i_{l,r}$ the union of the two integer intervals to the left and to the right of site $i$ of length $l$ and $r$, respectively; that is $\Delta^i_{l,r} = \{i-l,\dotsc, i-1,i+1,i+r\}$. We will also denote by $X_{\Delta^i_{l,r}}$ the random vector $\{X_j\colon j\in \Delta^i_{l,r}\}$ and similarly
$x_{\Delta^i_{l,r}}$ will denote an element of $A^{\Delta^i_{l,r}}$.

We say  $\P$ is a \emph{variable range} Markov random field
on $A$ if for any $i\in\Z$ there exist two integers $l_i$ and $r_i$
such that for all $L\geq l_i$ and all $R\geq r_i$ we have
 \begin{equation}\label{probs}
  \P(X_i = x_i \talque X_k = x_k, \tab k \in \Delta^i_{L,R}) \; =\; \P(X_j = x_j \talque X_k = x_k, \tab k \in \Delta^{i}_{l_i,r_i})\,,
\end{equation}
for all $x_i\in A$ and all $x_{\Delta^i_{L,R}}$ for which  $\P(X_{\Delta^i_{L,R}} = x_{\Delta^i_{L,R}}) > 0$.

Observe that the size of the neighborhood $\Delta^i_{l_i,r_i}$ may depend on the specific site $i$, for that reason we call our model a variable range Markov random field.
\begin{remark}\label{min}
If $(l_i,r_i)$ satisfies equation (\ref{probs}) then any pair $(l,r)$ with $l\geq l_i$ and $r\geq r_i$ will also satisfy equation (\ref{probs}). For this reason in the sequel we assume $l_i$ and $r_i$ are the minimal integers satisfying (\ref{probs}), calling the set
$\Delta^i_{l_i,r_i}$ the \emph{basic neighborhood} of site $i$.
\end{remark}

\begin{remark}\label{equal_neig}
Note that if $X_i$ and $X_j$ are not conditionally independent given all the remaining variables then $[i;j] \subset \Delta^i_{l_i,r_i}\cap \Delta^j_{l_j,r_j}$, where $[i;j]$ denotes the integer interval $\{i,i+1,\dotsc, j-1,j\}$. On the other hand, if there exists $\ell$ such that $r_i \leq \ell-i$ for any $i\leq \ell$ and  $l_j \leq j- \ell$ for any $j> \ell$,  then
$X_i$ is independent of $X_j$ for any $i\leq \ell$ and any $j>\ell$.
\end{remark}

In what follows we will focus on the problem of identifying the  basic neighborhood of a given site $i\in\Z$. Without loss of generality we will take $i=0$
and we will simply write $\Delta^0_{l_0,r_0} = \Delta_{l_0,r_0}$.
We  will assume we have an independent sample of  size $n$ of $(X_{-L_0},\dotsc,X_0,\dotsc,  X_{R_0})$, with $L_0\geq l_0$ and $R_0\geq r_0$. We will denote by $x^{(i)}_{j}$ the value taken by  the $j$-th variable in the $i$-th observation. Our goal is to estimate
the basic  neighborhood $\Delta_{l_0,r_0}$ (by estimating the parameters $l_0$ and $r_0$) and the conditional probabilities given by (\ref{probs}), based on this sample.

Given two sequences $w = (w_{-l},\dotsc,w_{-1}) \in A^l$ and $v= (v_1,\dotsc,v_r) \in A^r$ and a symbol $a\in A$ we will denote by $p(a|w,v)$ and $p(w,a,v)$ the conditional  (respectively joint) probability given by
\[
p(a|w,v) =  \P(X_0 = a \talque X_{-l:-1}=w, X_{1:r}=v)
\]
and
\[
p(w,a,v) =  \P(X_0 = a, X_{-l:-1}=w, X_{1:r}=v)\,,
\]
where  $X_{i:j}$ represents the sequence $X_i,\dotsc, X_{j}$.
The operator $N_n(w,a,v)$ will denote the number of occurrences of the event
\[
\{X_{-l:-1}=w\}\cap \{X_0=a\}\cap \{X_{1:r}\}
\]
in the sample. That is
\[
N_n(w,a,v) \;=\; \sum_{i=1}^n\mathbf{1}\{x^{(i)}_{-l:r}=wav\},
\]
where $wav$ is the concatenation of $w$, $a$ and $v$; that is
$wav=(w_{-l},\dotsc,w_{-1},a,v_1,\dotsc,v_r)\in A^{l+r+1}$.
Given $w$ and $v$, the maximum likelihood estimator of the conditional distribution $\{p(\cdot|w,v)\colon a\in A\}$
is given by
\begin{equation}\label{trans}
\hat p_n(a|w,v) = \frac{N_n(w,a,v)}{N_n(w,v)},\qquad \text{ for }a\in A\,,
\end{equation}
where $N_n(w,v) = \sum_{a\in A} N_n(w,a,v)$. If  $N_n(w,v)=0$ we adopt the convention $\hat p_n(a|w,v) = 1/|A|$ for all $a\in A$.

For any pair of integers $(l,r)$, with $l\leq L_0$ and $r\leq R_0$ we denote by
\begin{equation}\label{hatp}
\hat \P_{l,r}(x_0^{(1:n)}| x_{\Delta_{l,r}}^{(1:n)}) = \prod_{w\in A^l}\prod_{v\in A^r}\prod_{a\in A} \hat p_n(a|w,v)^{N_n(w,a,v)} \,.
\end{equation}

In order to estimate the neighborhood $\Delta_{l_0,r_0}$; that is, to estimate $l_0$ and $r_0$,  we propose to use a penalized  maximum (conditional) likelihood criterion.

\begin{definition}\label{def:est}
Given a constant $c>0$, the empirical neighborhood  of  site 0 is the set of indices
$\Delta_{\hat l_n,\hat r_n} = \{-\hat l_n,\dotsc,-1,+1,\hat r_n\}$,
where
\begin{equation}\label{argmax}
(\hat l_n,\hat r_n) \;=\; \underset{0\leq l\leq L_0,0\leq r\leq R_0}{\arg\max}\bigl\{ \,\log \hat\P_{l,r}(x_0^{(1:n)}| x_{\Delta_{l,r}}^{(1:n)})  - c\,|A|^{l+r}\log_{|A|} n \,\bigr\}\,.
\end{equation}
\end{definition}

We prove the following consistency result for the neighborhood estimator.

\begin{theorem}\label{main}
The estimator given by  (\ref{argmax})  satisfies $(\hat l_{n},\hat r_{n}) = (l_0,r_0)$ and
 therefore $\Delta^0_{\hat l_n,\hat r_n}=\Delta^0_{l_0,r_0}$ eventually almost surely
as $n\to \infty$.
\end{theorem}

The proof of Theorem~\ref{main} is given in Section~\ref{sec:proof}.

\section{Variable dependence windows for \snps maps and disease association}\label{sec:snps}

In this paper we model data from 2,062 individuals (72.4\% of them females) in which 868 are affected by rheumatoid arthritis (cases) and 1,194 are not affected (controls). This data is the initial batch of whole genome association data for the North American Rheumatoid Arthritis Consortium (\narac) provided by the Genetic Analysis Workshop 16 (GAW16). For all individuals in this case-control study, information from 545,080 \snp-genotype from the Illumina 550K chip are available. The genotypes are in the format X$_{-}$X, where X is a base (A,C,G,T). Each record has information about \snp name, chromosome, and \snp position in basepairs. Only genotypes from the 22 autosomal chromosomes were used in our analysis. For each \snp the scores 0, 1 and 2 were assigned for their three possible \snp genotypes. Thus, for instance, a \snp in which their genotypes are GG, AA, and GA, score 0 was assigned for the homozygote with highest frequency, score 1 for the heterozygote, and score 2 for the homozygote with lowest frequency.

In order to remove potential genotype errors, we excluded those \snps with minor allele frequency (\maf) lower than 1\% as well as those not in Hardy-Weinberg equilibrium (\hwe). The \hwe was checked by using the Chi-square test available in the genetics library of the R pa\-cka\-ge \citep{R2005}.
 The significance level considered was 10$^{-4}$. At~the end, a total of 43,616 \snps were removed with 501,464 remaining for the analysis. No procedure was used to impute the missing genotypes that remained in the data set.

 We apply the model and estimators described in Section~\ref{sec:model} to this data set.
We obtained in this way neighborhoods for each one of the $501,464$
\snps.  Considering the entire data set we obtained neighborhoods with mean
$2.22418$ \snps and
standard deviation $0.714481$ \snp. The mean size of the left part of the
neighborhoods was $1.11432$ and that of the right part was $1.10985$ \snp.

\begin{figure}
\centering\includegraphics[width=14cm]{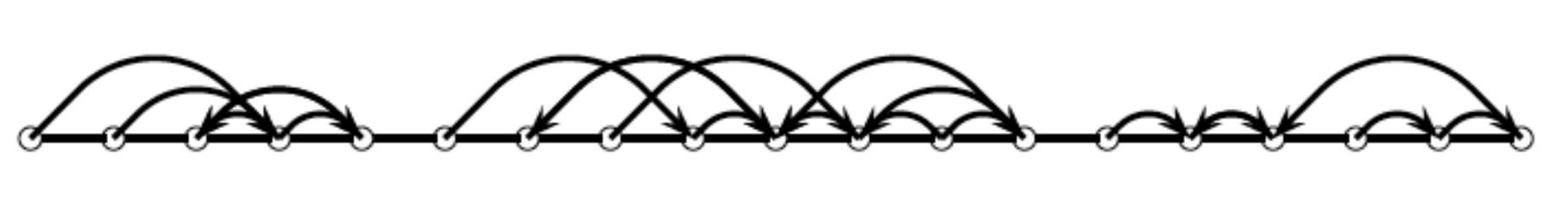}
  \caption{Left and right neighborhoods for the first 19 \snps in the sample.}
  \label{fig:intervals}
\end{figure}

An interesting property of the estimated neighborhoods can be observed
in Figure~\ref{fig:intervals}, where we represent the length of the left an right
parts of each neighborhood by an arrow.
In this representation we can see some  regions (between two adjacent \snps) that are not crossed by
any arrow; that is by the neighborhoods of the adjacent \snps. This points divide the set
of \snps into (probabilistically) independent blocks of different sizes (see Remark~\ref{equal_neig}).
We illustrate how these blocks are obtained in Figure~\ref{fig:blocks}. From now on we will called these independent blocks of ``influence windows''.

\begin{figure}
\centering
  \includegraphics[width=10cm]{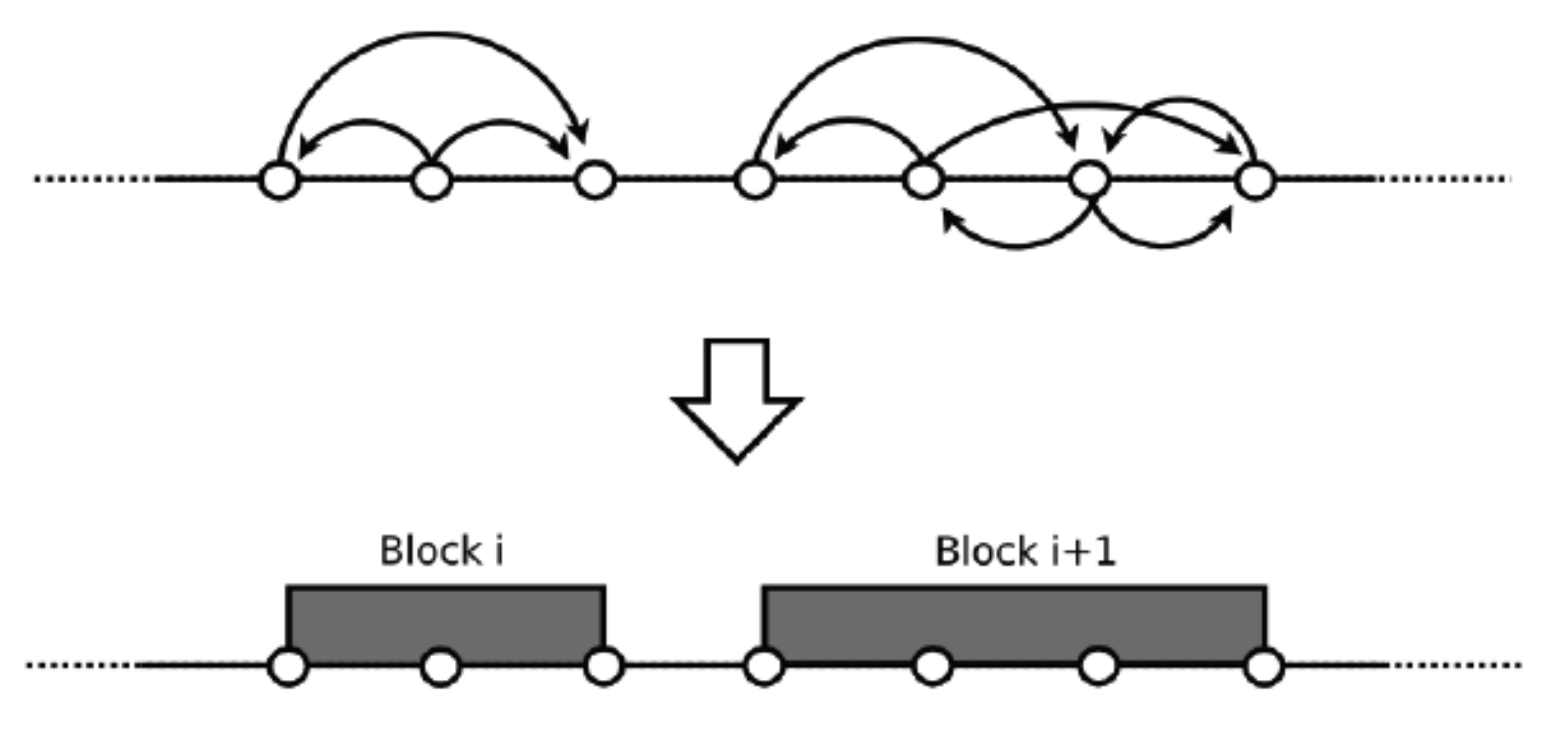}
  \caption{Independent blocks obtained by the identification of intermediate points not included in any neighborhood.}
  \label{fig:blocks}
\end{figure}

Analyzing the neigborhoods previously determined by the algorithm we obtained a total of  $48,697$ influence windows.
The mean size of these influence windows was $10.27$ \snps, the smallest window has only $1$ \snp and the biggest one has $83$ \snps.
The standard deviation of the sizes of the influence windows was $5.94$ \snps.

In order to test the association of every influence window with the rheumatoid arthritis, we perform a Chi-square test of independence between the observed genotype frequencies in that window and the response variable indicating the presence/absence of the disease.    In Figure~\ref{fig:chi} we show the scores, corresponding to minus the logarithm
in base 10 of the $p$-value,  calculated for each window in the 22 autosomes. We can observe a region of high association in the sixth chromosome, where 22 windows had a $p$-value smaller that $10^{-16}$.
These results are compatible with previous studies about  rheumatoid arthritis, as for example \citet{irigoyen2005,amos2009}.
We can observe also some windows in other chromosomes  that exhibited a small $p$-value and can therefore be associated with the disease. A list with the influence windows that had a $p$-value smaller that $10^{-4}$, as well as
the program written in C to perform this analysis is available and can be requested from the authors.

\begin{figure}
\centering
  \includegraphics[width=15cm]{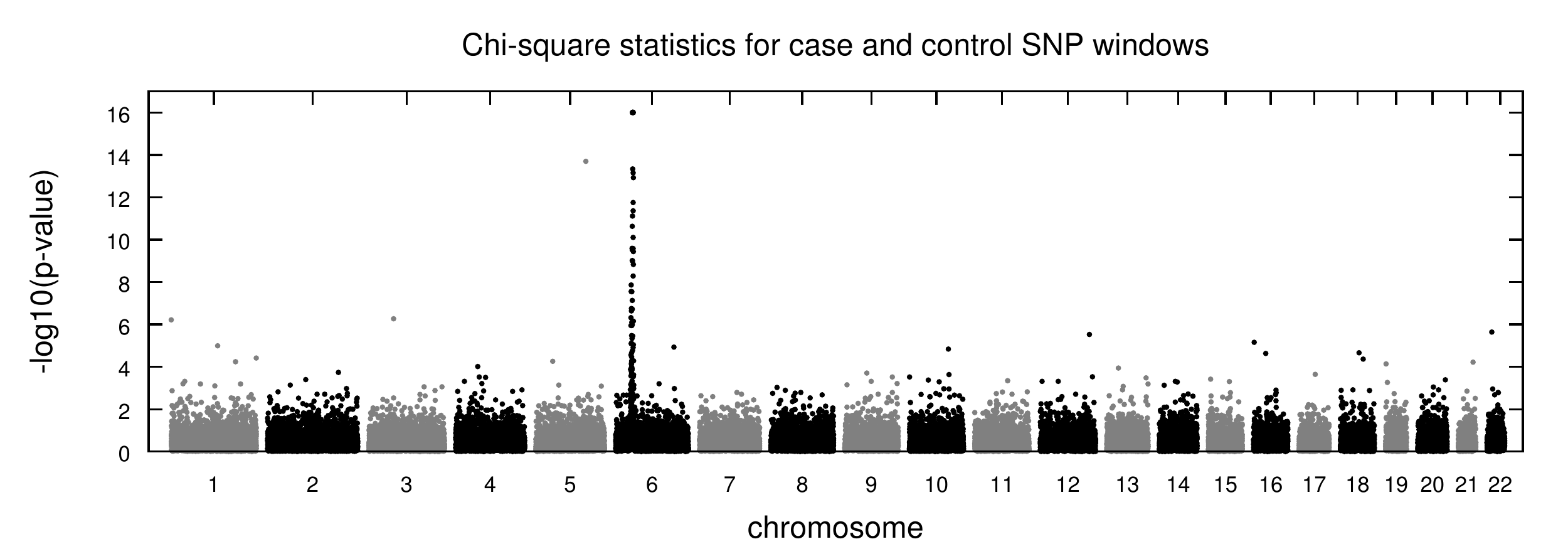}
  \caption{Chi-square statistics for each influence window in the 22 autosomes.}
  \label{fig:chi}
\end{figure}

\section{Proof of Theorem~\ref{main}}\label{sec:proof}

 We begin by giving some basic definitions and stating some results that will be useful in the proof of Theorem~\ref{main}.  From now on we simply write $\log$ for the logarithm in base $|A|$.

\begin{definition}\label{KL}
The \emph{K\"ullback-Leibler divergence}
between the two probability distributions $P$ and $Q$ over $A$
 is defined by
 $$D(P;Q) = \sum_{a\in A} P(a)\log\frac{P(a)}{Q(a)}$$
where, by convention, $P(a) \log\frac{P(a)}{Q(a)}=0$ if $P(a)=0$ and $P(a) \log\frac{P(a)}{Q(a)}=+\infty$ if $P(a) > Q(a)=0$.
\end{definition}

The following lemma was taken from \citet[Lemma~6.3]{csiszar2006}. We include it here for completeness, but we omit its proof.
\begin{lemma}\label{KLTV}
For any two probability distributions  $P$ and $Q$ over $A$ we have
\[
D(P;Q)  \;\leq\; \sum_{a\in A\colon Q(a)>0} \frac{[P(a)-Q(a)]^2}{Q(a)}\,.
\]
\end{lemma}

Now we prove a result showing an upper bound for the deviation of the empirical conditional probabilities
from their true values.

\begin{proposition}\label{prop:typicality}
For any $\delta>0$ and for any triple $(w,a,v)\in A^{l+r+1}$ with $0 \leq l\leq L_0$ and $0\leq r\leq R_0$ we have
\[
\bigl| \hat p_n(a|w,v) - p(a|w,v) \bigr|\;<\; \sqrt{\frac{\delta\log n}{N_n(w,v)}}
\]
eventually almost surely as $n\to\infty$.
\end{proposition}

\begin{proof}
Define, for fixed $(w,a,v)\in A^{l+r+1}$,  the random variables
\[
Y_i = \mathbf{1}\{x^{(i)}_{-l:r}=wav\} - p(a|w,v) \mathbf{1}\{x^{(i)}_{-l:-1}=w\}\mathbf{1}\{x^{(i)}_{1:r}=v\}\, ,\quad i=1,2,\dotsc,n
\]
and
\begin{equation}\label{mart}
Z_n\; =\; \sum_{i=1}^n Y_i\; =\; N_n(w,a,v) - p(a|w,v)N_n(w, v)\,.
\end{equation}
The variables $\{Y_i\colon i=1,2,\dotsc,n\}$ are i.i.d and a direct calculation gives, for any $i=1,\dotsc,n$,  $\E(Y_i)=0$
and
\[
\E(Y_i^2)\,=\,p(a|w,v)(1-p(a|w,v))p(w,v)\,\leq\, \frac{p(w,v)}{4}\,,
\]
where $p(w,v)=\sum_{a'\in A}p(w,a',v)$.
Now, by the Law of the Iterated Logarithm we have that for any $\epsilon>0$
\[
|Z_n| \;<\; (1+\epsilon) \frac{p(w,v)}{4} \sqrt{2n\log\log n}
\]
eventually almost surely as $n\to\infty$. In particular we have
\[
|Z_n| \;<\; \sqrt{2p(w,v)^2n\log\log n}
\]
eventually almost surely as $n\to\infty$. Dividing both sides of the inequality by $N_n(w,v)$
we obtain that
\[
\bigl| \hat p_n(a|w,v) - p(a|w,v) \bigr|\;<\; \sqrt{\frac{2p(w,v)^2n\log\log n}{N_n(w,v)^2}}\,.
\]
By the Strong Law of Large Numbers we have that $n/N_n(w,v)\to 1/p(w,v)$
almost surely, therefore  we have
\[
\bigl| \hat p_n(a|w,v) - p(a|w,v) \bigr|\;<\; \sqrt{\frac{4p(w,v)\log\log n}{N_n(w,v)}}
\]
eventually almost surely as $n\to\infty$.
Now, for any $\delta>0$ we have that
\[
 4p(w,v)\log\log n <  \delta  \log n
\]
eventually as $n\to\infty$,  and this concludes the proof of Proposition~\ref{prop:typicality} .
\end{proof}

\begin{proof}[Proof of Theorem~\ref{main}]

Denote by
\[
\pml_{l,r}(x_0^{(1:n)}| x_{\Delta_{l,r}}^{(1:n)})  \;=\;\log \hat\P_{l,r}(x_0^{(1:n)}| x_{\Delta_{l,r}}^{(1:n)})  - c|A|^{l+r}\log n\,.
\]

We will divide the proof in two cases.

\noindent\emph{(a) Overestimation}.
We have to prove that simultaneously for all pairs $(l,r)\neq (l_0,r_0)$, with $l_0\leq l \leq L_0$, $r_0\leq r\leq R_0$  we will have
\[
\pml_{l,r}(x_0^{(1:n)}| x_{\Delta_{l,r}}^{(1:n)}) \; <\;  \pml_{l_0,r_0}(x_0^{(1:n)}| x_{\Delta_{l_0,r_0}}^{(1:n)})
\]
eventually almost surely as $n\to\infty$.

Observe  that
\begin{align}\label{eq:pml}
\pml_{l_0,r_0}(&x_0^{(1:n)}| x_{\Delta_{l_0,r_0}}^{(1:n)}) - \pml_{l,r}(x_0^{(1:n)}| x_{\Delta_{l,r}}^{(1:n)})
\; = \notag\\
&c\,(|A|^{l+r} - |A|^{l_0+r_0})\log n \; - \sum_{wav\in A^{l+r+1}} N(wav) \log \frac{\hat p(a|w,v)}{\hat p(a|w_{-l_0:-1},v_{1:r_0})}\,,
 \end{align}
 where by an abuse of notation we write $N(wav) = N_n(w,a,v)$ and $\hat p(a|w,v) = \hat p_n(a|w,v)$.  As these empirical  probabilities are the maximum likelihood  estimators we have that
\begin{align*}
  \sum_{wav\in A^{l+r+1}} N(wav) \log \hat p(a|w_{-l_0:-1},v_{1:r_0}) \;&\geq\; \sum_{wav\in A^{l+r+1}} N(wav) p(a|w_{-l_0:-1},v_{1:r_0})\\
 &= \;  \sum_{wav\in A^{l+r+1}} N(wav) p(a|w,v)\,.
\end{align*}
Therefore, (\ref{eq:pml}) can be upper-bounded by
\[
c\, \Bigl(1 - \frac{1}{|A|}\Bigr)|A|^{l+r}\log n\; -\sum_{wav\in A^{l+r+1}} N(wav) \log \frac{\hat p(a|w,v)}{p(a|w,v)}\,.
\]
Now observe that
\begin{align*}
\sum_{wav\in A^{l+r+1}} N(wav) \log \frac{\hat p(a|w,v)}{p(a|w,v)}
 =  \sum_{w\cdot v\in A^{l+r}} N(w\!\cdot\! v)D(\hat p(\cdot|w,v) \,;\,p(\cdot|w,v) )\,,
\end{align*}
where $D$ denotes the \emph{Kullback-Leibler divergence} (see Definition~\ref{KL} in the the Appendix). Therefore, by Lemma \ref{KLTV} and Proposition~\ref{prop:typicality} we have   that for any $\delta >0$
\begin{align*}
\sum_{w\cdot v\in A^{l+r}} N(w\!\cdot\! v)&D(\hat p(\cdot|w,v) \,;\,p(\cdot|w,v) ) \\
&\leq\; \sum_{w\cdot v\in A^{l+r}} N(w\! \cdot \! v) \sum_{a\in A} \frac{[\,\hat p(a|w,v) - p(a|w,v)\,]^2}{p(a|w,v)}\\
&\leq\; \sum_{w\cdot v\in A^{l+r}} N(w\!\cdot \! v) \sum_{a\in A} \frac{\delta\log n }{N(w\!\cdot\! v) p(a|w,v)}\\
&\leq \;  \frac{\delta |A|^{l+r+1}\log n}{p_{\min}}\,,
\end{align*}
eventually almost surely as $n\to\infty$, where
\[
p_{\min}\;=\; \min\{\,p(a|w,v)\colon p(a|w,v)>0, a\in A, w\in A^{l_0}, v\in A^{r_0}\,\}\,.
\]
Then if we take $\delta< c \,p_{\min}(|A|-1)/|A|^2$ we have that  eventually almost surely as $n\to\infty$
\[
\pml_{l_0,r_0}(x_0^{(1:n)}| x_{\Delta_{l_0,r_0}}^{(1:n)})  \; >\; \pml_{l,r}(x_0^{(1:n)}| x_{\Delta_{l,r}}^{(1:n)})
\]
simultaneously for all pairs
$(l,r)\neq (l_0,r_0)$, with $l_0\leq l \leq L_0$, $r_0\leq r\leq R_0$.
This completes the proof of part (a).

\noindent\emph{(b) Underestimation.}
We have to prove that simultaneously for all pairs $(l,r)$ with $l<l_0$ or $r<r_0$
 we have
\[
\pml_{l_0,r_0}(x_0^{(1:n)}| x_{\Delta_{l_0,r_0}}^{(1:n)})  \; >\; \pml_{l,r}(x_0^{(1:n)}| x_{\Delta_{l,r}}^{(1:n)})
\]
eventually almost surely as $n\to\infty$.

First consider the case $l\leq l_0$ and $r\leq r_0$.  In this case we have that
\begin{align*}\label{eq:pml_under}
\pml_{l_0,r_0}(x_0^{(1:n)}| &x_{\Delta_{l_0,r_0}}^{(1:n)}) - \pml_{l,r}(x_0^{(1:n)}| x_{\Delta_{l,r}}^{(1:n)})
\; = \notag\\
&\sum_{wav\in A^{l_0+r_0+1}} N(wav) \log \frac{\hat p(a|w,v)}{\hat p(a|w_{-l:-1},v_{1:r})} - c\,(|A|^{l_0+r_0} - |A|^{l+r})\log n \\
=\;& n \;\biggl[\; \sum_{wav\in A^{l_0+r_0+1}} \frac{N(wav)}{n} \log \frac{\hat p(a|w,v)}{\hat p(a|w_{-l:-1},v_{1:r})} - c\,(|A|^{l_0+r_0} - |A|^{l+r}) \frac{\log n}{n}\Biggr] \,.
 \end{align*}
By the Strong Law of Large Numbers we have that
\begin{align*}
 \sum_{wav\in A^{l_0+r_0+1}} \frac{N(wav)}{n} \log &\frac{\hat p(a|w,v)}{\hat p(a|w_{-l:-1},v_{1:r})} \\
\longrightarrow\; &\sum_{wav\in A^{l_0+r_0+1}} p(wav) \log \frac{p(a|w,v)}{p(a|w_{-l:-1},v_{1:r})}
\end{align*}
almost surely as $n\to\infty$, where $p(wav) = \P(X_0=a, X_{\Delta_{l_0,r_0}}=wv)$. The Log-Sum inequality and the minimality of $(l_0,r_0)$ (see Remark~\ref{min})
implies that
\[
\sum_{wav\in A^{l_0+r_0+1}} p(wav) \log \frac{p(a|wv)}{p(a|w_{-l:-1},v_{1:r})} \;>\; 0
\]
As the number of pairs $(l,r) \neq (l_0,r_0)$ satisfying $l\leq l_0$ and $r\leq r_0$ is finite, this implies
that simultaneously for all such pairs we will have
\[
\pml_{l_0,r_0}(x_0^{(1:n)}| x_{\Delta_{l_0,r_0}}^{(1:n)})  \; >\; \pml_{l,r}(x_0^{(1:n)}| x_{\Delta_{l,r}}^{(1:n)})
\]
eventually almost surely as $n\to\infty$.

Now consider the case $(l,r)$ with $l_0 < l \leq L_0$ and $r<r_0$. We will prove that simultaneously for all $l_0< l \leq L_0$,
\[
\pml_{l,r}(x_0^{(1:n)}| x_{\Delta_{l,r}}^{(1:n)})  \;<\; \pml_{l_0,r}(x_0^{(1:n)}| x_{\Delta_{l_0,r}}^{(1:n)})
\]
eventually almost surely as $n\to\infty$.
The proof of this fact follows the same arguments of part (a), by observing that
 \begin{align*}
 \pml_{l_0,r}(x_0^{(1:n)}| x_{\Delta_{l_0,r}}^{(1:n)}) -  \pml_{l,r}(x_0^{(1:n)}| x_{\Delta_{l,r}}^{(1:n)})
\; \geq \;&c\, \Bigl(1 - \frac{1}{|A|}\Bigr)|A|^{l+r}\log n\\
 & -\sum_{wav\in A^{l+r+1}} N(wav) \log \frac{\hat p(a|w,v)}{p(a|w,v)}\\
 \;\geq \; & \Bigl[c\, \Bigl(1 - \frac{1}{|A|}\Bigr) -\frac{\delta |A|}{p_{\min}}\Bigr]   |A|^{l+r} \log n\\
 \;>\; &0
 \end{align*}
by Proposition~\ref{prop:typicality}, for a sufficiently small $\delta$, eventually almost surely as $n\to\infty$
and simultaneously for all $l_0< l \leq L_0$. This fact, combined with what was proved for the pair $(l_0,r)$ before
implies that
\[
\pml_{l,r}(x_0^{(1:n)}| x_{\Delta_{l,r}}^{(1:n)})  \;< \;\pml_{l_0,r_0}(x_0^{(1:n)}| x_{\Delta_{l_0,r_0}}^{(1:n)})
\]
 eventually almost surely as $n\to\infty$, simultaneously for all $r< r_0$ and $l_0\leq l \leq L_0$.

By observing that the same proof applies to the case $(l,r)$ with $l< l_0$ and $r_0<r\leq R_0$, this finishes the proof of part (b).

\end{proof}

\section*{Discussion}

In this paper we presented a method based on the theory of variable range Markov random fields to estimate the extent of dependence among \snps allowing variable windows along the genome. We proposed an estimator based on the idea of penalized maximum likelihood to find the conditional dependence neighborhood of each \snp in the sample and we proved its consistency. A major advantage of our method is that it is adaptive for the extent of dependencies among \snps and it is not necessary to specify a window size to capture the dependency pattern which is a problem in most sliding-windows approaches.

The core of our method is to find the basic neighborhoods of long sequences of \snps, without taking into account the disease status, and then to test the association of each independent block with the disease. Therefore, the method can be used for association with many different types of trait data, such as quantitative traits. It could also be applied to other platform types, like multiallelic markers  (e.g. microsatellites), as well as to other data sequences, like nucleotide or aminoacid sequences.  
In our analysis we considered genotypic data sequences, on the level of individuals, but another option is to use haplotype data, on the level of chromosomes, by phasing allele data. The challenge in the latter is to estimate  the phase of the data, but good haplotype-phasing computer programs are now available. The flexibility of our method to handle genotype or haplotype data may be useful to assess different disease models.

A reasonably large sample size is required to attain consistency of our neighborhood estimator. Rare long \snps blocks, which may be present in the population, are expected to be observed in low frequency and may bias the findings. In this direction, an open question is how to obtain lower bounds for the sample size to guarantee a given level of precision for the neighborhood estimator.
Beyond the sample size, the results are dependent of the density of the \snps covering the genome and also of the size of the alphabet being modeled. For biallelic markers, as \snps, these problems become less severe.


\section*{Acknowledgments}

This work is part of USP's project \emph{Mathematics, computation, language and the brain} (no 11.1.9367.1.5)
and CNPq project \emph{Stochastic modeling of the brain activity} (no 480108/2012-9) .
The data analyzed in this work was gathered with the support of grants from the
National Institutes of Health (NO1-AR-2-2263 and RO1-AR-44422), and the National
Arthritis Foundation. The Genetic Analysis Workshop was supported by NIH grant R01
GM031575 from the National Institute of General Medical Sciences.
F.L. is partially supported by a CNPq fellowship (grant 302162/2009-7).

\bibliographystyle{dcu}
\bibliography{./references}

\end{document}